\documentclass[review]{elsarticle}

\usepackage{graphicx}

\usepackage{amssymb}
\usepackage{amsthm}
\usepackage{amsmath}
\usepackage{amsfonts}

\newtheorem{theorem}{Theorem}[section]
\newtheorem{proposition}[theorem]{Proposition}
\newtheorem{definition}[theorem]{Definition}
\newtheorem{corollary}[theorem]{Corollary}
\newtheorem{lemma}[theorem]{Lemma}

\usepackage{hyperref}
\usepackage{pgf,tikz}
\usepackage{mathrsfs}
\usetikzlibrary{arrows}
\usepackage{yhmath}

\journal{Computational Geometry}

\begin{document}

\begin{frontmatter}

\title{Recognition of the Spherical Laguerre Voronoi Diagram}

\author{Supanut Chaidee\corref{cor1}}
\ead{schaidee@meiji.ac.jp}
\cortext[cor1]{Corresponding author}

\author{Kokichi Sugihara}

\address{Graduate School of Advanced Mathematical Sciences, Meiji University\\
	4-21-1 Nakano, Nakano-ku, Tokyo 164-8525, Japan}

\begin{abstract}
In this paper, we construct an algorithm for determining whether a given tessellation on a sphere is a spherical Laguerre Voronoi diagram or not. For spherical Laguerre tessellations, not only the locations of the Voronoi generators, but also their weights are required to recover. However, unlike the ordinary spherical Voronoi diagram, the generator set is not unique, which makes the problem difficult. To solve the problem, we use the property that a tessellation is a spherical Laguerre Voronoi diagram if and only if there is a polyhedron whose central projection coincides with the tessellation. We determine the degrees of freedom for the polyhedron, and then construct an algorithm for recognizing Laguerre tessellations.
\end{abstract}

\begin{keyword}
\texttt{spherical Voronoi diagram, Laguerre distance, recognition problem, projective transformation}
\MSC[2010] 51N15 \sep 68Q25 \sep 68U05  
\end{keyword}

\end{frontmatter}

%\linenumbers

\section{Introduction}

The Voronoi diagram is one of the most useful representations for tessellations in computational geometry. The fundamental concepts, generalizations, and applications have been studied widely, e.g., \cite{BAurenhammer,BOkabe}. Here we are concerned with weighted Voronoi diagrams called Laguerre Voronoi diagrams, which are also known as a power diagrams. This class of Voronoi diagrams are important because the boundaries are linear instead of curved.

The Laguerre Voronoi diagram was introduced by \cite{Imi,Aurenhammer}. Briefly, for a set $S$ of $n$ spheres $s_i=(\textbf{x}_i, r_i)$ in $\mathbb{R}^d$, where $\textbf{x}_i$ is the position of the center and $r_i$ is the radius, which is equivalent to the generator weight, the Laguerre distance of $\textbf{x}\in \mathbb{R}^d$ from $s_i$ is 
\begin{equation}\label{LagDisOr}
d_{\text{L}}(\textbf{x},s_i)=\|\textbf{x}-\textbf{x}_i\|^2-r_i^2.
\end{equation}
In \cite{Aurenhammer2,Aurenhammer3} the correspondence between the Laguerre Voronoi diagram in $\mathbb{R}^d$ and a polyhedron in $\mathbb{R}^{d+1}$ was investigated. Then using the correspondence, an algorithm for constructing the Laguerre Voronoi diagram was presented in \cite{Imi,Aurenhammer}, and a robust version was presented in \cite{Sugihara1}. The Laguerre Voronoi diagram can be extended for tessellations on a sphere as defined in \cite{Sugihara2}. Applications of the spherical Laguerre Voronoi diagram were also studied in \cite{Mach,Chaidee3} for modeling objects with spheres.

It is sometimes necessary to consider the inverse of the above problem: the determination of whether or not a given tessellation is the Voronoi diagram. If it is, we can recover the corresponding generators. This problem is known as the Voronoi recognition problem, and was studied in \cite{BLoeb,JAsh,JAurenhammer,JHartvigsen,JSchoenberg,JAloupis}. On the other hand, if the tessellation cannot be represented by a Voronoi diagram,  we approximate it with the best fitting Voronoi diagram; this is called the Voronoi approximation problem. Examples of the Voronoi approximation problems are found in \cite{JHonda1,JHonda2,JSuzuki,JEvan,Chaidee1}, for unweighted Voronoi diagrams. These approximations have a number of useful applications, for instance, if we have a tessellation that is found in the real world, and we can approximate the tessellation with a Voronoi diagram, we can use this Voronoi diagram as a model for understanding the pattern's formation.

The inverse problems were widely studied for the case of ordinary Voronoi diagrams, whereas relatively little work has been done for Laguerre diagrams. Duan et al. \cite{Duan} considered the inverse problem for the planar Laguerre Voronoi diagram, and presented an algorithm for recovering the generators and their weights from a tessellation. Lautensack \cite{Lau} and Lyckegaard et al. \cite{Lyckegaard} used Laguerre Voronoi diagrams to study the relation between structures and their physical properties. Recently, Spettl et al. \cite{Duan2} fitted Laguerre Voronoi diagrams to tomographic image data. In the case of the spherical Laguerre Voronoi diagram, Chaidee and Sugihara \cite{Chaidee3} provided a framework for approximating the weights of spherical Laguerre Voronoi diagrams when the location of the Voronoi generators are known.

In this paper, we focus on the spherical Laguerre Voronoi diagram recognition problem. Our goal is to judge whether or not a given spherical tessellation is a spherical Laguerre Voronoi diagram. Remark that in the case of the ordinary Voronoi diagram recognition problem, the generator positions are unique whereas there are many sets of the generating circles generating the same Laguerre Voronoi diagram. With this reason, the Laguerre Voronoi diagram recognition problem is more difficult than the ordinary Voronoi diagram recognition problem. By the nonuniqueness property of generating circles, for each spherical Laguerre Voronoi diagram, there is a class of polyhedra whose projections coincide with the spherical Laguerre Voronoi diagram. Therefore, for a given tessellation, if we find these polyhedra, we can judge it is a  spherical Laguerre Voronoi diagram, and can recover the generators and their weights. Otherwise, we judge that the given tessellation is not a spherical Laguerre Voronoi diagram.

This paper is organized as follows. In Section 2, we provide the definitions and theorems which are necessary for our study. The recognition problem is also mathematically defined in this section. In Section 3, we focus on the properties of projective transformations which transform a polyhedron associated with a spherical Laguerre Voronoi diagram to other polyhedra in the same class. In Section 4, we give algorithms for constructing a polyhedron from a given tessellation and for judging whether a given tessellation is a spherical Laguerre Voronoi diagram. Finally, we summarize our research in Section 5, and give suggestions for future work.

\section{Preliminaries}

In this section, we provide the necessary definitions and theories on the spherical Laguerre Voronoi diagram. Then we state the problem considered in this paper.

We first present some fundamental definitions from spherical geometry.

Let $U$ be a unit sphere whose center is located at the origin $O(0, 0, 0)$ of a Cartesian coordinate system. 

For two distinct points $p, p_i \in U$ with position vectors $\textbf{x}, \textbf{x}_i$, respectively, the \textit{geodesic arc} $\wideparen{e}_{p,p_i}$ of $p, p_i$ is defined as the shortest arc between $p$ and $p_i$ of the great circle passing through $p$ and $p_i$. The \textit{geodesic arc length}, also called the \textit{geodesic distance}, is given as follows,
\begin{equation}
\tilde{d}(p, p_i) = \arccos(\textbf{x}^{\text{T}}\textbf{x}_i)\leq \pi. \label{GeoDist}
\end{equation}

In this paper, we focus on spherical tessellations. We define the spherical polygon as follows.

Let $(q_1, ..., q_m)$ be a sequence of distinct vertices on the sphere such that the geodesic arcs $\wideparen{e}_{q_i,q_i+1}$ ($i=1, ..., m$; $q_{m+1}$ is read as $q_1$) do not intersect except at the vertices. The left area enclosed by the collection of these geodesic arcs is called the \textit{spherical polygon} $Q(q_1, ..., q_m)$. $Q(q_1, ..., q_m)$ is abbreviated as $Q$ hereafter.

The spherical polygon $Q$ is said to be \textit{convex} if and only if no geodesic arc joining the two points in $Q$ goes outside of $Q$.

$\mathcal{T}$ is said to be a \textit{spherical tessellation} if $\mathcal{T}$ is a decomposition of $U$ into a countable number of spherical polygons whose interiors are pairwise disjoint. $\mathcal{T}$ is said to be \textit{convex} if all of these spherical polygons are convex.

Next, we consider a special case of a spherical tessellation. Let $G=\{p_1, ..., p_n\}$ be a set of points on the sphere $U$. Assignment of each point on $U$ to the nearest point in $G$ with respect to the geodesic distance forms a tessellation, which is called the spherical Voronoi diagram on the sphere $U$. Algorithms for constructing the spherical Voronoi diagram were provided in \cite{Renka,Sugihara2}.

We can generalize the spherical Voronoi diagram to the spherical Laguerre Voronoi diagram. The necessary definitions and theorems were originally presented in \cite{Sugihara1,Sugihara2}. We briefly introduce those definitions and theorems as follows.

For a given sphere $U$ and point $p_i \in G$, the spherical circle $\tilde{c}_i$ centered at $p_i$ is defined by
\begin{equation}\label{SC}
\tilde{c}_i=\{p\in U|\tilde{d}(p_i, p)=r_i\},
\end{equation}
where $0\leq r_i <\pi/2$. $r_i$ is called the radius of the spherical circle $\tilde{c}_i$.

Following \cite{Sugihara2}, we define the Laguerre proximity, the distance measured from an arbitrary point $p$ on the sphere $U$ to a spherical circle  $\tilde{c}_i$, as follows,
\begin{equation}\label{RealCircle}
\tilde{d}_L(p, \tilde{c}_i)=\frac{\cos \left(\tilde{d}(p_i, p)\right)}{\cos \left(r_i\right)}.
\end{equation}
Let $\tilde{c}_i$ and $\tilde{c}_j$ be two circles. The Laguerre bisector of $\tilde{c}_i$ and $\tilde{c}_j$ is defined by
\begin{equation}
B_L(\tilde{c}_i, \tilde{c}_j)=\{p\in U|\tilde{d}_L(p, \tilde{c}_i)=\tilde{d}_L(p, \tilde{c}_j)\}. \label{LagBi}
\end{equation}
For a set $\tilde{G}=\{\tilde{c}_1, ..., \tilde{c}_n\}$ of $n$ spherical circles on $U$, we define the region
\begin{equation}
\tilde{R}(\tilde{G}, \tilde{c}_i)=\{p\in U|\tilde{d}_L(p, \tilde{c}_i)<\tilde{d}_L(p, \tilde{c}_j), j\neq i\}.
\end{equation}
The regions $\tilde{R}(\tilde{G}, \tilde{c}_1), ..., \tilde{R}(\tilde{G}, \tilde{c}_n)$, together with their boundaries constitute a tessellation, which is called the spherical Laguerre Voronoi diagram of $U$. Figure \ref{SLVD} shows an example of the spherical Laguerre Voronoi diagram. This is a stereo diagram. If the right diagram is viewed with the left eye and the left diagram with the right eye, the sphere can be seen; the upper pair represents the front hemisphere and the lower pair represents the rear hemisphere.

\begin{figure}[h]
	\begin{center}
		\includegraphics[scale=0.7]{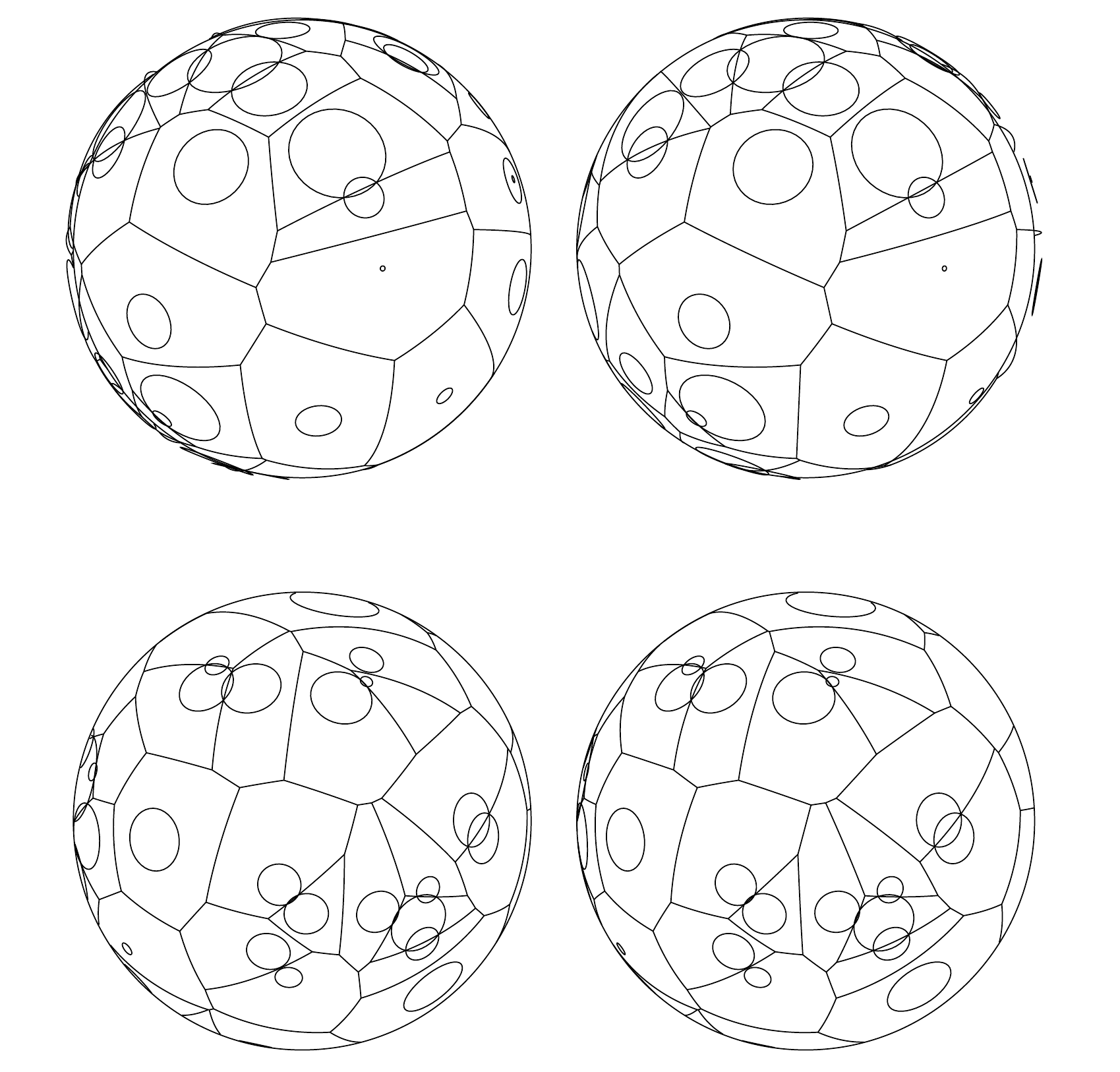}
		\caption{Stereographic images of a spherical Laguerre Voronoi diagram.} \label{SLVD}
	\end{center}
\end{figure}

In \cite{Sugihara1,Sugihara2}, an algorithm for constructing the spherical Laguerre Voronoi diagram and its dual, i.e., the spherical Laguerre Delaunay diagram, was proposed. For a spherical circle  $\tilde{c}_i$ of $U$, let $\pi(\tilde{c}_i)$ be the plane passing through $\tilde{c}_i$, and let $H(\tilde{c}_i)$ be the halfspace bounded by $\pi(\tilde{c}_i)$ and containing $O$. Let $\ell_{i,j}$ be the line of intersection of $\pi(\tilde{c}_i)$ and $\pi(\tilde{c}_j)$.

The Laguerre bisector shown in (\ref{LagBi}) is characterized by the following theorems.

\begin{theorem}[\cite{Sugihara2}]\label{SugiThm1}
	The Laguerre bisector $B_L(\tilde{c}_i,\tilde{c}_j)$ is a great circle, and it crosses the geodesic arc connecting the two centers $p_i$ and $p_j$ at right angles.
\end{theorem}

\begin{theorem}[\cite{Sugihara2}]\label{SugiThm2}
	The bisector $B_L(\tilde{c}_i,\tilde{c}_j)$ is the intersection of $U$ and the plane containing $\ell_{i,j}$ and $O$.
\end{theorem}

For a set $\tilde{G}$ of spherical circles, the spherical Laguerre Voronoi diagram is constructed by the following process. For circles $\tilde{c}_i$, we construct the planes $\pi(\tilde{c}_i)$ and the half spaces $H(\tilde{c}_i)$, and the intersection of all halfspaces. We finally project the edge of the resulting polyhedron onto $U$, with the center of the projection at $O$. Then we have the spherical Laguerre Voronoi diagram.

Throughout this paper, the term \textit{spherical polygon} is abbreviated as \textit{polygon}. Let $\mathcal{T}$ be the given convex spherical tessellation on the unit sphere $U$ having $n$ polygons. We assume that all vertices of $\mathcal{T}$ are of degree 3.

The main concern of this paper is to judge whether or not a given tessellation $\mathcal{T}$ is a spherical Laguerre Voronoi diagram.

\section{Polyhedron Transformation}

From the definition and theorems of the spherical Laguerre Voronoi diagram in the previous section, the following proposition state the correspondence between the spherical Laguerre Voronoi diagram and a polyhedron.
\begin{proposition}\label{twoWay}
	$\mathcal{L}$ is a spherical Laguerre Voronoi diagram if and only if there is a convex polyhedron $\mathcal{P}$ containing the center of the sphere whose central projection coincides with $\mathcal{L}$.
\end{proposition}
\begin{proof}
	We firstly prove the necessity of the condition. Let $\mathcal{L}$ be a spherical Laguerre Voronoi diagram. Hence, the generating circles exist. Therefore, we get the intersection of halfspaces with boundary planes passing through the spherical circles including the center of the sphere. The projection of this polyhedron onto $U$ coincides with $\mathcal{L}$ because of Theorem \ref{SugiThm2}.
	
	Inversely, let $\mathcal{L}$ be a given spherical tessellation and $\mathcal{P}$ be a convex polyhedron containing the center of the sphere such that its central projection coincides with $\mathcal{L}$. If there exists a plane $P_i$ such that $P_i$ does not intersect $U$, shrink the polyhedron $\mathcal{P}$ so that $P_i$ intersects $U$ for all $i$. Let $\tilde{c}_i$ be the intersection of the plane containing the face $P_i$ with $U$ for all $i$. Then, by Theorem \ref{SugiThm2}, $\mathcal{L}$ is the spherical Laguerre Voronoi diagram for $\tilde{c}_i$'s as the generators.
	
\end{proof}

Therefore, for each spherical Laguerre Voronoi diagram $\mathcal{L}$, there is a class of polyhedra whose central projections coincide with $\mathcal{L}$. To solve the spherical Laguerre Voronoi recognition problem, we will construct an algorithm to find those polyhedra.

To find these polyhedra, we study the transformation of the polyhedron that preserves the projection onto the sphere $U$. Let $P^3(\mathbb{R})$ be the three-dimensional projective space. The following properties are the requirements for the transformation.

\begin{definition}(Projection Preservation Property)\label{PPP}
	Let $f$ be a transformation from $P^3(\mathbb{R})$ to $P^3(\mathbb{R})$. $f$ is said to be the \textit{projection preserving mapping with respect to the origin $O$} if $f$ satisfies the following properties:
	\begin{enumerate}
		\item $f(O)=O$;
		\item For any point $v\in P^3(\mathbb{R})$, $v$ and $f(v)$ are on the same line passing through $O$.
	\end{enumerate}
\end{definition}

Let $\textbf{v}_a=(t_a, x_a, y_a, z_a) \in P^3(\mathbb{R})$ be a homogeneous coordinate representation of a vertex of the polyhedron $\mathcal{P}$ in the projective space. We define a map $f:P^3(\mathbb{R})\rightarrow P^3(\mathbb{R})$ by
\begin{equation}\label{map}f(\textbf{v}_a)=\begin{pmatrix}
\alpha & \beta & \gamma & \delta \\ 
0 & \eta & 0 & 0 \\ 
0 & 0 & \eta & 0 \\ 
0 & 0 & 0 & \eta
\end{pmatrix}\textbf{v}_a \end{equation}
where $\alpha, \beta, \gamma, \delta, \eta \in \mathbb{R}$ and $\alpha\neq 0, \eta\neq 0$.

\begin{theorem}\label{transMapThm}
	The mapping $f$ defined by equation (\ref{map}) is a projection preserving mapping.
\end{theorem}
\begin{proof}
	
	We now prove that the map $f$ satisfies the conditions in Definition \ref{PPP}. It is easy to verify that $f(O)=O$. It remains to be shown that $f(\textbf{v}_a)=\textbf{v}_a'$, where $\textbf{v}_a'$ lies on the line passing through the origin $L_a$ of $\textbf{v}_a$.
	
	For $\textbf{v}_a=(t_a, x_a, y_a, z_a) \in P^3(\mathbb{R})$, we have
	\begin{equation}
	f(\textbf{v}_a)=\begin{pmatrix}
	\alpha & \beta & \gamma & \delta \\ 
	0 & \eta & 0 & 0 \\ 
	0 & 0 & \eta & 0 \\ 
	0 & 0 & 0 & \eta
	\end{pmatrix}\begin{pmatrix}
	t_a \\ 
	x_a \\ 
	y_a \\ 
	z_a
	\end{pmatrix} = \begin{pmatrix}
	\varLambda \\ 
	\eta x_a \\ 
	\eta y_a \\ 
	\eta z_a
	\end{pmatrix} \label{transMap}
	\end{equation}
	where $\varLambda=\alpha t_a + \beta x_a + \gamma y_a + \delta z_a$.
	
	Note that for the transformed point $\textbf{v}_a'=(\varLambda, \eta x_a, \eta y_a, \eta z_a)$ in the projective space $P^3(\mathbb{R})$, 
	$$\left(\frac{\eta x_a}{\varLambda}, \frac{\eta y_a}{\varLambda}, \frac{\eta z_a}{\varLambda}\right)= \frac{\eta t_a}{\varLambda}\left(\frac{x_a}{t_a}, \frac{y_a}{t_a}, \frac{z_a}{t_a}\right)$$
	is a point in the space $\mathbb{R}^3$ which implies that $\textbf{v}_a'$ lies on the same line as $\textbf{v}_a$, which concludes the proof.
\end{proof}

Since the transformation is a projective map, it preserves the planarity of faces of a polyhedron. Also, by the reason that $f$ satisfies the projection preservation property, the transformed point $f(v)$ is on the same line as $v$ passing through $O$. However, the transformation as defined in (\ref{map}) does not guarantee the convexity of the projected polyhedron because some vertices may be mapped to the other side of the origin. Note that if $\alpha=\eta=1$ and $\beta=\gamma=\delta=0$, the map $f$ is an identity map. The set of the transformations of the form (\ref{map}) form a continuous group, and hence for each convex polyhedron $\mathcal{P}$, there exists a positive constant $\epsilon$ such that for any $-\epsilon \leq \alpha-1, \eta-1, \beta, \gamma, \delta \leq \epsilon$, the transformation (\ref{map}) maps $\mathcal{P}$ to a convex polyhedron.

Note that the transformation is uniquely determined if we fix the five parameters $\alpha, \beta, \gamma, \delta$, and $\eta$. However, the homogeneous coordinate representation itself has one degree of freedom to represent each point. Hence, the choice of the polyhedron transformed from $\mathcal{P}$ by (\ref{map}) has four degrees of freedom.

\section{Algorithms}

\subsection{Recognition procedure}

In the previous section, the existence of the class of polyhedra whose projections coincide with the given spherical Laguerre Voronoi diagram is proved by Theorem \ref{transMapThm}. In this section, we propose an algorithm for constructing a polyhedron with respect to a given tessellation.  Let $\mathcal{T}$ be a given spherical tessellation. If a polygon $i$ is adjacent to a polygon $j$, then there exists a geodesic arc $\wideparen{e}_{i,j}$ which is a tessellation edge partitioning polygons $i$ and $j$. The tessellation vertex that is the intersection of edges $\wideparen{e}_{i,j}$, $\wideparen{e}_{j,k}$, and $\wideparen{e}_{i,k}$ is denoted by $v_{i,j,k}$, as shown in Figure \ref{polygons}. During the algorithm, $\ell_{i,j}$ is defined as the line intersecting $\pi(\tilde{c}_i)$ and $\pi(\tilde{c}_j)$.

\begin{figure}[h]
	\begin{center}
		\definecolor{uququq}{rgb}{0.25098039215686274,0.25098039215686274,0.25098039215686274}
		\begin{tikzpicture}[line cap=round,line join=round,>=triangle 45,x=1.0cm,y=1.0cm]
		\clip(2.8160640470324694,3.5067554088250974) rectangle (9.718902215139163,6.965008986114204);
		\draw [shift={(6.342008502719588,3.617125273275333)}] plot[domain=0.035538821603734644:3.1108262356690495,variable=\t]({1.*3.168612589062205*cos(\t r)+0.*3.168612589062205*sin(\t r)},{0.*3.168612589062205*cos(\t r)+1.*3.168612589062205*sin(\t r)});
		\draw [shift={(11.743458897258435,5.290836630761076)}] plot[domain=2.8644792148819374:3.229205627611263,variable=\t]({1.*5.449456835313343*cos(\t r)+0.*5.449456835313343*sin(\t r)},{0.*5.449456835313343*cos(\t r)+1.*5.449456835313343*sin(\t r)});
		\draw [shift={(8.476872026504891,7.202718578526532)}] plot[domain=3.9767765086173488:4.666125717960333,variable=\t]({1.*3.2218106032764178*cos(\t r)+0.*3.2218106032764178*sin(\t r)},{0.*3.2218106032764178*cos(\t r)+1.*3.2218106032764178*sin(\t r)});
		\draw [shift={(5.341385632169549,8.590198620619178)}] plot[domain=4.419859512679194:4.9646990170726895,variable=\t]({1.*3.902136057996358*cos(\t r)+0.*3.902136057996358*sin(\t r)},{0.*3.902136057996358*cos(\t r)+1.*3.902136057996358*sin(\t r)});
		\draw (5.814606881363096,4.746316083031945) node[anchor=north west] {$v_{i,j,k}$};
		\draw (6.479352421394042,6.623284324326742) node[anchor=north west] {$\wideparen{e}_{i,j}$};
		\draw (4.497349383026773,6.021848919580811) node[anchor=north west] {$\text{polygon }i$};
		\draw (7.268554244414403,5.447751487777876) node[anchor=north west] {$\text{polygon }j$};
		\draw (5.659213233104137,4.067183854156534) node[anchor=north west] {$\text{polygon }k$};
		\draw (4.647708234213256,5.14703378540491) node[anchor=north west] {$\wideparen{e}_{i,k}$};
		\draw (7.559202352642416,4.5182604077159825) node[anchor=north west] {$\wideparen{e}_{j,k}$};
		\begin{scriptsize}
		\draw [fill=uququq] (6.314903792454812,4.81400408777157) circle (1.5pt);
		\end{scriptsize}
		\end{tikzpicture}
		\caption{Three adjacent polygons corresponding to a tessellation vertex.} \label{polygons}
	\end{center}
\end{figure}
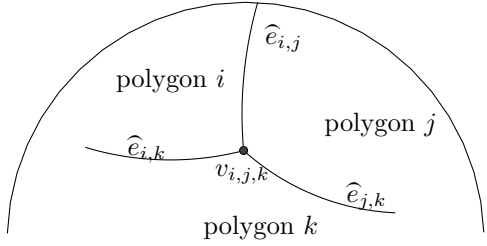

The overview of the recognition procedure is to try to construct a set of planes composing a polyhedron whose central projection coincides with $\mathcal{T}$. The planes are chosen sequentially. Suppose that the chosen sequence of the first three planes of polygons is $(i, j, k)$. The choice of the first two planes of polygons $i, j$ has the degrees of freedom, while the third plane of the polygon $k$ can be constructed uniquely up to the first two planes of polygons $i, j$.

The following algorithm is for the construction of the first three planes of the polyhedron corresponding to the three polygons $i, j, k$.
\\
\\
\textbf{Algorithm 1: Plane Construction with Three Adjacent Sites}
\\ \textbf{Input:} Tessellation edges $\wideparen{e}_{i,j}$, $\wideparen{e}_{j,k}$, $\wideparen{e}_{i,k}$, and degree-three tessellation vertex $v_{i,j,k}$.
\\ \textbf{Output:} The three planes $P_i, P_j, P_k$ with respect to polygons $i, j, k$.
\\ \textbf{Procedure:} 
\begin{enumerate}
	\item select a spherical circle  $\tilde{c}_i$ with center $p_i\in U$ and radius $r_i$ in polygon $i$.
	\item construct a plane $P_i:=\pi(\tilde{c}_i)$ of the spherical circle  $\tilde{c}_i$;
	\item construct a plane $P_{i,j}$, $P_{i, k}$, $P_{j,k}$ passing through $\wideparen{e}_{i,j}$, $\wideparen{e}_{i,k}$, and $\wideparen{e}_{j,k}$, respectively;
	\item find the line $\ell_{i,j}$ by intersecting $P_i$ and $P_{i,j}$, and $\ell_{i,k}$ by intersecting $P_i$ and $P_{i,k}$;
	\item construct a geodesic arc  $\wideparen{e}_{i,j}^c$ such that $\wideparen{e}_{i,j}^c$ passes through $p_i$ and is perpendicular to  $\wideparen{e}_{i,j}$;
	\item choose a point $q_j$ in polygon $j$ on the arc $\wideparen{e}_{i,j}^c$;
	\item construct the plane $P_j$ passing through $\ell_{i,j}$ and $q_j$;
	\item find the line $\ell_{j,k}$ by intersecting the planes $P_j$ and $P_{j,k}$;
	\item construct the plane $P_k$ passing through the line $\ell_{i,k}$ and $\ell_{j,k}$.
\end{enumerate}
\textbf{end Procedure}
\\
\\
In step 8 of Algorithm 1, $\ell_{j,k}$ is constructed from the intersection of $P_j$ and $P_{j,k}$ which is used in step 9 of the Algorithm 1. The following lemma guarantees the co-planarity of $\ell_{i,k}$ and $\ell_{j,k}$.

\begin{lemma}
	$\ell_{i,k}$ generated in step 4 of Algorithm 1 and $\ell_{j,k}$ generated in step 8 are co-planar.
\end{lemma}
\begin{proof}
	From Algorithm 1, suppose that the construction sequence is $(i, j, k)$. Then there exists an intersection between planes $P_i$ and $P_{i,k}$ and $P_{i, j}$, written as $\ell_{i,k}$ and $\ell_{i,j}$, respectively which are obviously coplanar. Since $P_{i, k}$ and $P_{i, j}$ pass through $\wideparen{e}_{i,k}$ and $\wideparen{e}_{i,j}$ such that $\wideparen{e}_{i,k}$ intersects $\wideparen{e}_{i,k}$ at $v_{i, j, k}$, there exists a line $\ell_{i, j, k}$ which is the intersection of $P_{i, k}$ and $P_{i,j}$. Therefore, there exists the intersection point of $\ell_{i, k}, \ell_{i,j}$ and $\ell_{i,j,k}$, says $V_{i,j,k}$.
	
	By step 7 of Algorithm 1, the plane $P_j$ is constructed through $\ell_{i,j}$, and also $V_{i, j, k}$. Then $\ell_{j, k}$ is constructed through the intersection of $P_j$ and $P_{j, k}$. Hence, $V_{i, j, k}$ is laid in the line $\ell_{j, k}$. Since $V_{i,j,k}\in \ell_{i,k}$ and $V_{i,j,k}\in \ell_{j,k}$, there exists the unique plane passing through $\ell_{i,k}$ and $\ell_{j,k}$ which implies that $\ell_{i,k}$ and $\ell_{j,k}$ are co-planar. 
\end{proof}

We next extend this process to construct all planes $P_1, ..., P_n$ corresponding to the polygons $1, ..., n$ of the given tessellation. 

Let $\mathcal{V}_i$ be the set of vertices of the $i$-th spherical polygon. Note that $\mathcal{V}_i$ is written as $\mathcal{V}_i=\{v_{i,j_1,k_1},...,v_{i,j_m,k_m}\}$ where $m$ is the number of vertices of the $i$-th spherical polygon. The set of spherical tessellation vertices is denoted by $\mathcal{V}=\cup_{i=1}^n \mathcal{V}_i$. 
\\
\\
\textbf{Algorithm 2: Construction of $n$ planes }
\\ \textbf{Input:} Spherical tessellation $\mathcal{T}$ where all vertices are of degree 3, and the set $\mathcal{V}$ of tessellation vertices.
\\ \textbf{Output:} The planes $P_1, ..., P_n$ with respect to the polygons $1, ..., n$, and $\text{Mark}(v)\in \{0, 1\}$ for $v\in \mathcal{V}$.
\\ \textbf{Comment:} $\mathbb{P}$ is the set of planes constructed in the  procedure.
\\ \textbf{Procedure:} 
\begin{enumerate}
	\item make $\mathbb{P}$ empty;
	\item set $\text{Mark}(v)=0$ for all $v \in \mathcal{V}$;
	\item choose an arbitrary vertex $v_{i, j, k}\in \mathcal{V}$ and employ Algorithm 1 to construct planes $P_i, P_j, P_k$;
	\item add the planes $P_i, P_j, P_k$ to $\mathbb{P}$;
	\item set $\text{Mark}(v_{i, j, k})=1$;
	\item \textbf{while} there exists vertex $v_{p,q,l}\in\mathcal{V}$ such that $\text{Mark}(v_{p,q,l})=0$ and exactly two planes $P_p, P_q$ are included in $\mathbb{P}$,\\
	\textbf{do} \\
	\text{\hspace{0.4cm}} apply steps 3, 4 of Algorithm 1, where $(i, j, k)$ are read as $(p, q, l)$, to find $\ell_{p,q}$ and $\ell_{p,l}$;\\
	\text{\hspace{0.4cm}} compute $\ell_{q,l}$ from the intersection of $P_q$ and $P_{q,l}$;\\
	\text{\hspace{0.4cm}} compute $\ell_{p, q, l}$ from the intersection of $P_{q, l}$ and  $P_{p, l}$;\\
	\text{\hspace{0.4cm}} compute $V_{p, q, l}$ from the intersection of $P_p$ and $\ell_{p, q, l}$;\\
	\text{\hspace{0.4cm}} choose a point $v'_{q, l}$ on the line $\ell_{q, l}$;\\
	\text{\hspace{0.4cm}} construct a plane $P_l$ from the point $v'_{q, l}, V_{p, q, l}$ and line $\ell_{p,l}$.\\
	\text{\hspace{0.4cm}} add $P_l$ to $\mathbb{P}$;\\
	\text{\hspace{0.4cm}} set $\text{Mark}(v_{p, q, l})=1$;\\
	\textbf{end while}
\end{enumerate}
\textbf{end Procedure}
\\
\\
In step 6 of Algorithm 2, for any arbitrary planes $P_p, P_q\in \mathbb{P}$ to construct a plane $P_l$, the choice of the point $v'_{q, l}$ from the line $\ell_{q, l}$ affects the uniqueness of the plane. That is, it takes a degree of freedom for each $v'_{q,l}$. However, if the given tessellation $\mathcal{T}$ is a spherical Laguerre Voronoi diagram, then a point $v'_{q, l}$ in step 6 of Algorithm 2 is arbitrarily chosen to obtain the unique plane $P_l$ passing through $\ell_{p, l}$ and $\ell_{q, l}$, which means that there is no degree of freedom in the choice of $v'_{q,l}$. For that purpose, we claim that for any point $v'_{q, l}$ on the line $\ell_{q, l}$ and $\ell_{p, l}$ are co-planar, which is proved by the following lemma.

\begin{lemma}\label{Coplanar2}
	If $\mathcal{T}$ is a spherical Laguerre Voronoi diagram, then the plane $P_l$ of the construction sequence $(p, q, l)$ constructed in the step 6 of Algorithm 2 does not depend on the choices of $v'_{q,l}$.
\end{lemma}
\begin{proof}
	Let $\mathcal{T}$ be a spherical Laguerre Voronoi diagram. We prove $V_{p,q,l}\in\ell_{q,l}$ to imply that any choices of $v'_{q,l}$ on $\ell_{q,l}$ gives the same plane $P_l$ passing through $V_{p,q,l}, \ell_{p,l}$ and $\ell_{q,l}$.
	
	In the step 6 of Algorithm 2, suppose that $v_{p, q, l}$ is the tessellation vertex of the adjacent polygons $p, q, l$. Suppose that there are planes $P_p, P_q \in \mathbb{P}$. By Theorem \ref{SugiThm2}, $\ell_{p,q}$ of the intersection of $P_p$ and $P_q$ is laid on the plane $P_{p,q}$ passing through the geodesic arc $\wideparen{e}_{p,q}$. Note that the line $\ell_{p,q,l}$ of the intersection of planes $P_{p, q}$ and $P_{q, l}$ is also included in $P_{p, l}$, and the intersection $V_{p, q, l}$ of $P_p$ and $\ell_{p, q, l}$ is on the lines $\ell_{p, q}$ and $\ell_{p, l}$.
	
	Using the fact that $\ell_{p,q}$ lays on the plane $P_{p,q}$ of $\wideparen{e}_{p,q}$ and Theorem \ref{SugiThm2}, it is implied that $\ell_{p,l}$ ,the intersection of $P_p$ and $P_{p,l}$, and $\ell_{q,l}$, the intersection of $P_q$ and $P_{q,l}$, intersect at $V_{p,q,l}$, which means that $V_{p,q,l}\in\ell_{q,l}$. Therefore, there exist a unique plane $P_l$ independent from the choice of $v'_{q,l}$.
\end{proof}

In Algorithm 1, we arbitrarily choose the first spherical circle  $\tilde{c}_i$ (i.e., with generator position $p_i$ and radius $r_i$), and the generator position with the choice of $q_j$ of the adjacent polygon $j$ which lies on the geodesic arc $\wideparen{e}_{i,j}^c$. This means that we have four degrees of freedom pursuing in Algorithm 1. This reflects the freedom in the choice of the polyhedron for representing the spherical Laguerre Voronoi diagram as shown in the following theorem.

\begin{theorem}\label{4degs}
	There are exactly four degrees of freedom in the choice of a polyhedron $\mathcal{P}$ with respect to the given spherical Laguerre Voronoi diagram.
\end{theorem}

\begin{proof}
	To prove this theorem, we will derive the lower and upper bounds for the degrees of freedom in the choice of the polyhedron.
	
	The lower bound is four as we have seen in Theorem \ref{transMapThm} and the discussion immediately after the theorem.
	
	For the upper bound, we show that if we have two polyhedra $\mathcal{P}$ and $\mathcal{P}'$ whose projections give the same spherical Laguerre Voronoi diagram, $\mathcal{P}$ is transformed to $\mathcal{P}'$ by the transformation (\ref{map}).
	
	Let $\mathcal{P}$ and $\mathcal{P}'$ be two different polyhedra whose central projections give the same spherical Laguerre Voronoi diagram. Assume that $P_i$ and $P_j$ are planes containing mutually adjacent faces $F_i, F_j$ of $\mathcal{P}$. Then there are vertices $v_{i,j,k_1}, v_{i,j,k_2}\in F_i \cap F_j$ for some $k_1, k_2$. Without loss of generality, we choose vertices $v_{i, m_1, m_2}\in F_i$ and $v_{j, n_1, n_2}\in F_j$ of the polyhedron $\mathcal{P}$. Note that the remaining polyhedron vertices are uniquely constructed by Algorithm 2 and Lemma \ref{Coplanar2}.
	
	On the polyhedron $\mathcal{P}'$, we consider the polyhedron vertices $v'_{i,j,k_1}, v'_{i,j,k_2}\in F'_i \cap F'_j$ and $v'_{i, m_1, m_2}\in F'_i$ and $v'_{j, n_1, n_2}\in F'_j$. Then there exists a transformation $f$ which transforms the vertices $v_{i,j,k_1}$, $v_{i,j,k_2}$, $v_{i, m_1, m_2}$, $v_{j, n_1, n_2}$ to the vertices $v'_{i,j,k_1}$, $v'_{i,j,k_2}$, $v'_{i, m_1, m_2}$, $v'_{j, n_1, n_2}$, respectively.
	
	Let $\mathcal{P}''$ be the polyhedron transformed by $f$ from $\mathcal{P}$. Since the first four points of polyhedron $\mathcal{P}''$ are $v'_{i,j,k_1}=f(v_{i,j,k_1})$ , $v'_{i,j,k_2}=f(v_{i,j,k_2})$, $v'_{i, m_1, m_2}=f(v_{i, m_1, m_2})$, $v'_{j, n_1, n_2}=f(v_{j, n_1, n_2})$, and all the other vertices are determined by Algorithm 2, $\mathcal{P}''$ is the same polyhedron as $\mathcal{P}'$. This means that the upper bound of the degrees of freedom is four.
	
	Thus, the degrees of freedom for choosing the polyhedron are exactly four.
\end{proof}

Note that in Algorithm 1, we choose 4 parameters arbitrarily, two for $p_i$, one for $r_i$ in step 1 and one for $q_j$ in step 6. Therefore, if the given tessellation $\mathcal{T}$ is a spherical Laguerre Voronoi diagram, the projection of the arrangement of the planes constructed by Algorithm 2 on the sphere gives $\mathcal{T}$ for any spherical circle $\tilde{c}_i$ and the point $q_j$ is chosen in Algorithm 1. 

We have the following corollary which is directly implied from Lemma \ref{Coplanar2}, Algorithm 2, and Theorem \ref{4degs}.

\begin{corollary}\label{uniquePoly}
	If $\mathcal{T}$ is a spherical Laguerre Voronoi diagram, then Algorithm 2 gives the unique polyhedron up to the choice of the first four degrees of freedom.	
\end{corollary}

We use the contraposition of Corollary \ref{uniquePoly} to verify that, if the constructed planes composing a polyhedron with respect to the tessellation $\mathcal{T}$ using Algorithm 2 are not uniquely constructed from the first choice of four degrees of freedom, then the given tessellation is not the spherical Laguerre Voronoi diagram.

The following lemma characterizes the properties of the polyhedron vertices and the given tessellation.

\begin{lemma}\label{ConsCheckLemma}
	Let $\mathcal{T}$ be a given tessellation  and $\mathbb{P}$ a set of planes constructed by Algorithm 2. $\mathcal{T}$ is a spherical Laguerre Voronoi diagram if and only if for all vertices $v_{i, j, k}\in \mathcal{V}$,
	\begin{enumerate}
		\item there exists the unique point $V_{i,j,k}$ of the intersection of the plane $P_i, P_j, P_k \in \mathbb{P}$; and
		\item there exists $t\in\mathbb{R}\backslash\{0\}$ such that $V_{i,j,k}=tv_{i,j,k}$.
	\end{enumerate}
\end{lemma}
\begin{proof}
	Firstly, let $\mathcal{T}$ be a spherical Laguerre Voronoi diagram and $v_{i,j,k}\in \mathcal{V}$. Without loss of generality, assume that $P_i, P_j$ are constructed sequentially. By Lemma \ref{Coplanar2}, $P_k$ is uniquely determined from $P_i$ and $P_j$ which is constructed from $\ell_{i,k}$ and $\ell_{j,k}$. Hence, $V_{i,j,k}\in P_i\cap P_j \cap P_k$ uniquely. In addition, since $v_{i,j,k}\in \wideparen{e}_{i,j}\cap \wideparen{e}_{i,k} \cap \wideparen{e}_{j,k}$ and $\ell_{i,j,k}:=P_{i,j}\cap P_{j,k}\cap P_{i,k}$ is a line passing through $O$ and $v_{i,j,k}$,  Algorithm 2 and Lemma \ref{Coplanar2} implies that $V_{i,j,k}$ is laid on the line $\ell_{i,j,k}$. That is, the condition 1 and 2 of the lemma are satisfied.
	
	Conversely, suppose that the conditions 1 and 2 of the lemma hold. Since $V_{i,j,k}$'s are formulated from the halfspaces intersection of the planes in $\mathbb{P}$ which are constructed from Algorithm 2, and the central projection is preserved from the condition 2, Proposition \ref{twoWay} implies that $\mathcal{T}$ is a spherical Laguerre Voronoi diagram.
\end{proof}

Lemma \ref{ConsCheckLemma} means that from the constructed polyhedron, the vertices selected in Algorithm 2, i.e., the vertices $v$ with $\text{Mark}(v)=1$ set in Algorithm 2, are located on the lines emanating from the origin and passing through the associated tessellation vertices. Therefore, the given tessellation $\mathcal{T}$ is not the spherical Laguerre Voronoi diagram if other vertices are not necessarily on the radial lines emanating from the origin and passing through the associated tessellation vertices.

As mentioned before, the uniqueness of the plane $P_l$ constructed from $P_p$ and $P_q$ is up to the choice of the point $v'_{q,l}$ in the step 6 of Algorithm 2. Lemma \ref{Coplanar2}, Corollary \ref{uniquePoly} together with Lemma \ref{ConsCheckLemma} lead to Algorithm 3 for checking whether the given tessellation is a spherical Laguerre Voronoi diagram or not.

From Algorithm 2, the set of tessellation vertices $\mathcal{V}$ is divided into two groups: marked vertices and unmarked vertices. The unmarked vertices are vertices $v\in\mathcal{V}$ such that $\text{Mark}(v)=0$. Remark that for the marked vertex $v_{p,q,l}$, the plane $P_l$ constructed from $P_p$ and $P_q$ in Algorithm 2 passes through the intersection point $V_{p,q,l}$ on the line $\ell_{p,q,l}$ passing through $v_{p,q,l}$. Therefore, it is sufficient to check the consistency in Algorithm 3 among the unmarked vertices.
\\
\\
\textbf{Algorithm 3: Spherical Laguerre Voronoi Diagram Recognition}
\\ \textbf{Input:} The tessellation $\mathcal{T}$, the set $\mathbb{P}$ and $\text{Mark}(v), v\in \mathcal{V}$ constructed from Algorithm 2.
\\ \textbf{Output:} ``true'' or ``false''.
\\ \textbf{Comment:} ``true'' means that $\mathcal{T}$ is a spherical Laguerre Voronoi diagram.
\\ \textbf{Procedure:} 
\begin{enumerate}
	\item choose an unmarked vertex $v_{p,q,l}$;
	\item \textbf{if} compute the intersection $V_{p,q,l}$ of planes $P_p, P_q, P_l\in\mathbb{P}$, there exists $t\in \mathbb{R}$ such that $V_{p,q,l}=tv_{p,q,l}$ \textbf{then}\\
	\text{\hspace{0.4cm}} $\text{Mark}(v_{p, q, l})=1$;\\
	\text{\hspace{0.0cm}} \textbf{else}\\
	\text{\hspace{0.4cm}} report ``false'' and terminate the process;\\
	\textbf{end if};
	\item \textbf{if} $\text{Mark}(v)=1$ for all $v\in \mathcal{V}$ \textbf{then}\\
	\text{\hspace{0.4cm}} report ``true'';\\
	\text{\hspace{0.0cm}} \textbf{else}\\
	\text{\hspace{0.4cm}} go to step 1;\\
	\textbf{end if}
\end{enumerate}
\textbf{end Procedure}

\subsection{Algorithm analysis}

In this section, we analyze the proposed algorithm. In Algorithm 1, the main operations are for the plane construction and the intersection of two planes, each of whose complexity is $O(1)$. Since Algorithm 1 is related to the first three polygons, the complexity of Algorithm 1 is $O(1)$.

The spherical Laguerre Voronoi diagram recognition problem is mainly considered by Algorithms 2 and 3. The following theorem shows the complexity of the problem.

\begin{theorem}
	For an $n$ cells tessellation, we can judge whether the given tessellation is a spherical Laguerre Voronoi diagram in $O(n \log n)$ time.
\end{theorem}

\begin{proof}
	
	We now consider the complexity of Algorithm 2. In Algorithm 2, the time complexity of steps 1 to 5 is $O(1)$. Most of the computation time is spent on step 6. We use a priority queue with a heap structure whose nodes contain vertices as follows.
	
	For each tessellation vertex $v\in \mathcal{V}$, we define the key function $k(v)=c$, where $c$ is the number of planes constructed around vertex $v$. Firstly, let $k(v)=0$ for all $v\in \mathcal{V}$. During the course of processing, the key value increases and $0 \leq k(v) \leq 3$. At the end of step 3 of Algorithm 2, we have $k(v_{i,j,k})=3$ and for the three vertices $v_{i,j,a_1}, v_{i,a_2,k}$ and $v_{a_3,j,k}$ adjacent to $v_{i,j,k}$, $k(v_{i,j,a_1})=k(v_{i,a_2,k})=$ $k(v_{a_3,j,k})=2$. 
	
	We store vertices with key values less than or equal to 2 in the heap according to the key function so that the root node has the largest key value. For each repetition of step 6 in Algorithm 2, the root node is deleted from the heap, and the vertex contained within is set to $v_{p,q,l}$. After we construct the plane, we update the key function. Then we remove nodes which have a key value of 3. The other nodes whose key values have been changed are relocated in the heap. We repeat this process until all nodes have been removed from the heap. Adding and removing a node from the heap requires $O(\log n)$ time. Since the spherical tessellation is planar, the number of vertices is $O(n)$. Hence, the complexity of step 6 in Algorithm 2 is $O(n \log n)$. 
	
	After we obtain the planes $P_1,..., P_n$ from Algorithm 2, the construction of the polyhedron can be performed using the intersection of halfspaces, including the origin, which requires $O(n\log n)$ time.
	
	In Algorithm 3, we have to check unmarked vertices. For each vertex, we pick the planes of the corresponding polygons and find the intersection of three planes at the vertex which take the time complexity $O(1)$. Therefore, the complexity of Algorithm 3 is $O(n)$ through the unmarked vertices.
	
	Summarizing the above considerations, we have proved that the complexity of the spherical Laguerre Voronoi  diagram recognition is $O(n \log n)$.
	
\end{proof}

\subsection{Interpretation of the generators}

Once the tessellation $\mathcal{T}$ is judged as a spherical Laguerre Voronoi diagram, by Algorithm 1, 2, 3, we can construct the associated set of generators in the following way.

First, shrink the polyhedron constructed by Algorithm 2 so that all the planes intersect with $U$. Next, collect all the intersections, i.e., circles, between the planes and $U$. The resulting set of circles is the set of generators. Note that the generator set is not unique. The set of circles obtained by the above procedure is an example of the generator set.

\section{Concluding Remarks}

We proposed an algorithm for determining whether a spherical tessellation is a spherical Laguerre Voronoi diagram. The algorithm utilizes the convex polyhedron corresponding to a given tessellation. The criterion for recognizating a spherical Laguerre Voronoi diagram is obtained from this polyhedron. If the given tessellation is a spherical Laguerre Voronoi diagram, we can recover the generating spherical circles.

The properties of the polyhedron presented in the algorithm can be applied to the spherical Laguerre Voronoi approximation problem. If we have a spherical tessellation which does not correspond exactly with a spherical Laguerre Voronoi diagram, employing Algorithm 2 can give us an approximate spherical Laguerre Voronoi diagram. An interesting problem is the determination of the spherical Laguerre Voronoi diagram which provides the best approximation to the spherical tessellation.

From the perspective of practical applications, we \cite{Chaidee3} recently presented an algorithm for finding the spherical Laguerre Voronoi diagram which most closely approximates a tessellation, using an example of a planar tessellation extracted from a photo taken from a curved surface with generators. For the case that a given spherical tessellation does not contain generators, we may apply the proposed algorithms for approximating the generators and their weights. One of our future areas for study is the application of the algorithms presented here to the analysis of polygonal patterns found in the real world.

\subparagraph*{Acknowledgements}

The first author acknowledges the support of the MIMS Ph.D. Program of the Meiji Institute for the Advanced Study of Mathematical Sciences, Meiji University, and the Development and Promotion of Science and Technology Talents Project (DPST) of the Institute for the Promotion of Teaching Science and Technology (IPST), Ministry of Education, Thailand. The authors also thank the reviewers who gave useful comments for improving the manuscript. This research is supported in part by the Grant-in-Aid for Basic Research [24360039]; and Exploratory Research [151512067] of MEXT.

%\appendix
%\section{Morbi eros magna
%%
%% Bibliography
%%

%% Either use bibtex (recommended), but commented out in this sample

\section*{References}

\bibliography{ref}

%% .. or use bibitems explicitely

\end{document}